\documentclass[sts]{ims}

\RequirePackage[OT1]{fontenc}
\usepackage{amsthm,amsmath,amssymb,amsfonts,natbib,bm,cases}
\usepackage[pdftex]{graphicx}
\usepackage[margin=3.0cm]{geometry}
\RequirePackage[colorlinks,citecolor=blue,urlcolor=blue]{hyperref}

\startlocaldefs

 \newcommand{\lp}{\left(}
 \newcommand{\rp}{\right)}
 \newcommand{\llp}{\left\{}
 \newcommand{\rrp}{\right\}}
 \newcommand{\lllp}{\left[}
 \newcommand{\rrrp}{\right]}

\numberwithin{equation}{section}
\theoremstyle{plain}
\newtheorem{theorem}{Theorem}[section]
\newtheorem{lemma}{Lemma}[section]
\newtheorem{proposition}{Proposition}[section]
\theoremstyle{definition}
\newtheorem{definition}{Definition}
\newtheorem{example}[definition]{Example}

\endlocaldefs

\begin{document}

\begin{frontmatter}
\title{Identification Problem for The Analysis of Binary Data with Non-ignorable Missing}
\runtitle{Identification Problem}

\begin{aug}
\author{\fnms{Kosuke} \snm{Morikawa}\thanksref{t1}\ead[label=e1]{morikawa@sigmath.es.osaka-u.ac.jp}
},
\author{\fnms{Yutaka} \snm{Kano}\thanksref{t1}\ead[label=e2]{kano@sigmath.es.osaka-u.ac.jp}
}
\runauthor{Morikawa and Kano}
\address{Division of Mathematical Science, Graduate School of Engineering Science, \\
Osaka University, Toyonaka, Osaka 560-8571, Japan.}
\affiliation{Osaka University}
\thankstext{t1}{\printead{e1,e2}}
\end{aug}
\begin{abstract}
When a missing-data mechanism is NMAR or non-ignorable, missingness is itself vital information and it must be taken into the likelihood, which, however, needs to introduce additional parameters to be estimated. The incompleteness of the data and introduction of more parameters can cause the identification problem. When a response variable is binary, it becomes a more serious problem because of less information of binary data, however, there are no methods to  briefly verify whether a mode is identified or not. Therefore, we provide a new necessary and sufficient condition to easily check model identifiability when analyzing binary data with non-ignorable missing by conditional models. This condition can give us what condition is needed for a model to have identifiability as well as make easily check the identifiability of a model 
\end{abstract}
\end{frontmatter}

\section{Introduction}
\label{sec:2.1}
In statistical analysis, it is an awkward situation to lose data which we supposed to get completely as initially scheduled. In recent scientific experiments, subjects were told that they could drop out anytime they wanted from the perspective of human participant protection. Thus, there are many dropouts in some experiments. For example, \citet{machin88} reported results of comparative trial of two dosages of depot medroxyprogesterone acetate (DMPA, 100mg and 150mg) in which subjects are missing over 40$\%$  at the endpoint. In DMPA trial test, 1151 subjects were divided into two dosages randomly and took DMPA in every quarter, over one year.  They reported the results of DMPA as a binary data: if subjects experience amenorrhea denote by 1, else denote by 0. The judge whether a subject experienced amenorrhea or not was based on her menstrual diary. Each woman generates a sequence according to whether or not she experienced amenorrhea in the successive reference periods. The number of women with each sequence is shown in Table \ref{tb:1.1} where ``$\times$'' means missingness. For example, ``$0 1\times\times$'' means amenorrhea is not absent at first period, but present next period, and the data can not be obtained from third period. This data has been analyzed by several authors by several approaches(e.g., \citealp{birmigham03,matsuyama04,willkins06}). 
 
In the analysis of repeated measure data, serial correlations of a response variable $\bm{Y}_i=[Y_{i1},\,\ldots,\,Y_{iT}]' $ may not be ignored and any statistical model for $Y_i$ has to take the correlations into account. There are largely two approaches by which one incorporates the serial correlations into the models: conditional models and marginal models. Conditional models describe the serial correlation by modeling $Y_t$, which is the response at time $t$, given not only covariates $X$, but also $Y_1,\,\ldots,\,Y_{t-1}$, which are responses recorded early in time. The approach is intuitive and facile, and the serial correlations of $Y_i$ are obtained easily from the conditional model. On the other hand, there exists several models marginal models to analyze categorical data with non-ignorable missingness. \citet{fay86}, \citet{baker88} and \citet{park94} among them have used log-linear models to analyze them. Marginal models are also developed, which describe the serial correlation by modeling $\bm{Y}$'s moments given covariates $\bm{X}$ (e.g., E[$Y_1\,|\,X$],\,E[$Y_1Y_2\,|\,X$]) by adopting a fully parametric approach or by modeling the limited number of lower-order moments only, where they do not model under conditioning on the responses recorded early in time \citep{fitzmaurice93, molenberghs94,molenberghs97,molenberghs05}. Furthermore, recently, more complicated models are being proposed. A hyblid models is one of these models, which retains advantageous features of the  selection and pattern-mixture model approaches simultaneously\citep{willkins06, yuan09}. 

However, there exists an essential problem ``unidetifiabilitiy of models''\citep{fitzmaurice95,matsuyama04}. If the model does not have identifiability,  any statistical inference is distorted and asymptotic properties are not guaranteed such as consistency and asymptotic normality. Unfortunately, there are no methods to verify identifiability easily. 

The likelihood of conditional models are prone to be simple, but we can learn only the direct effects of covariates to the responses since responses recorded early in time are conditioned. Marginal models and the hybrid model can give us total effects of covariates to the response variables, which are often what we are most interested in. However, The likelihood of  marginal models and hybrid models is prone to be complicated, in addition, its parameter space may be restricted or demanded equality constraints. For example, in \citet{molenberghs97}, equality constraints are placed on the coefficients of the missing-data mechanism over time. This requires the probability of missingness is invariant throughout the experiment, which is an unnatural assumption since, in many cases, subjects would more weight on the response variables to decide to drop out the study at the start of experiment than at the endpoint.

Throughout this paper, let $\bm{Y}=[Y_1,\ldots,Y_T]'$ be a random binary variable and $\bm{M}_i=[M_{1},\ \ldots,\ M_{T}]' $ be the missing indicator, which takes 0(1) when corresponding component of $Y_{t}$ is observed(missing). The reason why we designate $\bm{Y}$  as a binary random variable is that the binary case is most difficult to become identifiable. 

\begin{table}[htbp]
\caption{Results of DMPA Trial}
\label{tb:1.1}
\begin{center}
{\small
\begin{tabular}{ccccc}
\hline
time & \shortstack{Amenorrhea\\ sequence} & & \multicolumn{2}{c}{DMPA(mg)} \\
\cline{4-5} 
 & & & 100 & 150 \\ \hline
1 & 0$\times \times \times$ & & 76 & 68 \\ 
  & 1$\times \times \times$ & & 23 & 31 \\
\cline{3-5}
 & & total & 99 & 99 \\ \hline 
2 & 0 0$ \times \times$ & & 43 & 39 \\ 
  & 0 1$ \times \times$ & & 14 & 27 \\
  & 1 0$ \times \times$ & & 3  & 6 \\
  & 1 1$ \times \times$ & & 8  & 15 \\
  \cline{3-5}
  & & total & 68 & 87 \\ \hline 
3 & 0 0 0$ \times$ & & 20 & 11 \\
  & 0 0 1$ \times$ & & 13 & 10 \\ 
  & 0 1 0$ \times$ & & 1  & 0  \\ 
  & 0 1 1$ \times$ & & 5  & 6  \\ 
  & 1 0 0$ \times$ & & 2  & 1  \\ 
  & 1 0 1$ \times$ & & 2  & 1  \\ 
  & 1 1 0$ \times$ & & 0  & 1  \\ 
  & 1 1 1$ \times$ & & 5  & 6  \\
  \cline{3-5}
  & & total & 48 & 36 \\ \hline 
4 & 0 0 0 0 & & 142 & 119 \\
  & 0 0 0 1 & & 49  & 36  \\
  & 0 0 1 0 & & 14  & 26  \\
  & 0 0 1 1 & & 41  & 44  \\
  & 0 1 0 0 & & 7   & 4   \\
  & 0 1 0 1 & & 8   & 12  \\
  & 0 1 1 0 & & 4   & 7   \\
  & 0 1 1 1 & & 32  & 48  \\
  & 1 0 0 0 & & 6   & 3   \\
  & 1 0 0 1 & & 7   & 6   \\
  & 1 0 1 0 & & 0   & 2   \\
  & 1 0 1 1 & & 10  & 12  \\
  & 1 1 0 0 & & 4   & 1   \\
  & 1 1 0 1 & & 4   & 3   \\
  & 1 1 1 0 & & 3   & 2   \\
  & 1 1 1 1 & & 30  & 28  \\
  \cline{3-5}
  & & total & 361 & 353 \\ \hline \\
\end{tabular}
}
\end{center}
\end{table}
\section{Identification Problem}
Typically, a categorical random variable $\bm{Y}$ has less information than a continuous random variable, and the lack of information can lead to ``identification problem'' as well as decrease of accuracy of estimation. To see this problem simply, assume the type of missing patterns is drop out. Denote focusing models of joint distribution function of $(\bm{M},\,\bm{Y})$ by $g$, a realized value of $\bm{M}$ whose number of observed elements is $t$  by $\bm{m}^{(t)} (t=1,\ \ldots,\,T)$, and observation(missing) part of $\bm{Y}$ by $\bm{Y}^{(t)}\left(\bm{Y}^{(-t)}\right)$ when the missing pattern is $\bm{m}^{(t)}$. Note that $\bm{Y}=[\bm{Y}^{(t)'},\,\bm{Y}^{(-t)'}]'$ always holds. In addition, let $g$ be one of the conditional models defined as follows: 
\begin{align*}
P(Y_t=1\mid y_1,\ldots,y_{t-1})=P(Y_t=1\mid \bm{h}^{((t-1)/p)})
\end{align*}
where \begin{align*}
\bm{h}^{(t/p)}:=
  \begin{cases}
    [y_{t-p},\,\ldots,\,y_{t}]' & \mathrm{if}~t-p\geq 1\\
    [y_{1},\,\ldots,\,y_{t}]' & \mathrm{otherwise}
  \end{cases}
  \qquad t=2,\,\ldots,\,T
\end{align*}
and
\begin{align*}
P(M_1=0)&=1,\\
P(M_t=1\mid M_{t-1}=0,\,y_1,\,\ldots,\,y_{t})&=P(M_t=1\mid M_{t-1}=0,\,y_{t-1},\,y_{t}),
\end{align*}
which we call AR($p$) model. In this model, it is assumed that $Y_t$ depends on the past own data until at most $p$ times and the missing-data mechanism depends on present data and only past nearest one data as with \citet{diggle94}. Note that ($Y_2,\,\ldots,\,Y_T$) may be missing and there are no covariates. The reason why considering the situation where there are no covariates at first is that this likelihood becomes so simple that we can study the identifiability easily.  

Modeling the relation between $\bm{Y}$ and missing indicator $\bm{M}$, we can introduce the limited number of parameters because of poor information of $\bm{Y}$. For example, in Table \ref{tb:1.1}, the total number of cells is $2+4+8+16=30$, thus, we can use at most $29$ parameters.  Let $T$ be the endpoint of the experiment and then we can use at most
\begin{align}
\sum_{t=1}^{T} 2^t -1 = 2^{T+1}-3
\label{2.1}
\end{align}
parameters. Because this condition is necessary not sufficient,  there would exist many models with parameters less than or equal to \eqref{2.1} but unidentified. 
\begin{example}
{\bf Logistic AR(1) model}\\
Suppose that the missing-data mechanism is given as
\begin{align}
P(M_t=1\mid M_{t-1}=0,\,y_{t-1},y_t;\tau_{t0},\tau_{tt-1},\tau_{tt})=\mathrm{expit}(\tau_{t0}+\tau_{tt-1}y_{t-1}+\tau_{tt}y_{t})
\label{Logistic}
\end{align}
and that the marginal distribution of $\bm{Y}$ is expressed in the form:
\begin{align}
P(Y_1=1;\theta_1)&=\theta_1,
\label{pi}\\
P(Y_t=1\mid y_{t-1};\theta_{t0},\theta_{tt-1})&=\mathrm{expit}(\theta_{t0}+\theta_{tt-1}y_{t-1})\nonumber,
\end{align}
where ``expit" is the inverse function of ``logit" function, each $\theta_{t0}$ and $\tau_{t0}$ are  intercepts in the model, and  $\theta_{tt-1}$, $\tau_{tt-1}$ and $\tau_{tt}$ are coefficients in the models. We call this model the Logistic AR(1) model. Note that when $\tau_{tt-1}=\tau_{tt}=0$ for $t=2,\,\ldots,\,T$, its mechanism is MCAR; when $\tau_{tt}=0$  for $t=2,\,\ldots,\,T$ and there exists $s$ such that $\tau_{ss-1}\neq 0$, it is MAR;  when there exists $s\,(s=2,\,\ldots,\,T)$ such that $\tau_{ss}\neq 0$, it is NMAR. In the Logistic AR(1) model, there are one parameter $\theta_1$ and five parameters $\lp\bm{\xi}_t :=[\theta_{t0},\theta_{tt-1},\tau_{t0},\tau_{tt-1},\tau_{tt}]'=[\bm{\theta}'_t,\bm{\tau}'_t]'\rp$ at each time $t~(2\leq t\leq T)$. Thus, the number of the parameters is
\begin{align}
1+5(T-1)=5T-4.
\label{2.2}
\end{align}
The relation between \eqref{2.1} and \eqref{2.2} is 
\begin{align*}
 \begin{cases}
    2^{T+1}-3<5T-4 & \mathrm{if}\quad T=2 \\
    2^{T+1}-3>5T-4 & \mathrm{if}\quad T\geq 3
  \end{cases} .
\end{align*}
Therefore, if $T=2$, the model does not have identifiability and if $T\geq 3$, the model meets the necessary condition. As we can see in a later section, however, the identifiability dose not hold for all $T$. To see this, we define ``identifiability '' explicitly at first.
\end{example}
\begin{definition}
Let $\Xi$ be a parameter space, $\bm{\xi}^*$  be a true value of the model and an interior point of $\Xi$, $P_{\bm{\xi}^*}$ be a probability measure of a probability function of complete data $(\bm{M},\,\bm{Y})$ prescribed by a true parameter $\bm{\xi}^*$, and denote a probability function of observed data $(\bm{M},\,\bm{Y}^{(t)})$ derived from $g$ by $g_t\,(t=1,\ldots,T)$, which is represented as
\begin{align*}
g_t\lp\bm{m}^{(t)},\,\bm{y}^{(t)}_i\,;\,\bm{\xi}\rp=
\begin{cases}
\sum_{\bm{y}^{(-t)}\in\{0,1\}^{\otimes(T-t)}}g\lp\bm{m}^{(t)},\,\bm{y}^{(t)}_i,\, \bm{y}^{(-t)}\,;\,\bm{\xi}\rp & \mathrm{if~}t=1,\ldots, T-1 \\
g\lp\bm{m}^{(t)},\,\bm{y}_i\,;\,\bm{\xi}\rp & \mathrm{otherwise} 
\end{cases}.
\end{align*}
Then, a parametric model $g$ is said to be identifiable, if
 \begin{align}
g_t\lp\bm{m}^{(t)},\,\bm{y}^{(t)}\ ; \ \bm{\xi}\rp= g_t\lp\bm{m}^{(t)},\,\bm{y}^{(t)}\ ; \ \bm{\xi}^*\rp\quad \mathrm{a.s.}~P_{\bm{\xi}^*}\quad \mathrm{for}\ \forall t,\,\forall\bm{y}^{(t)}\ \Rightarrow\ \bm{\xi}= \bm{\xi}^*. 
\label{2.3}
\end{align}
Since $\bm{m}^{(t)}$ and $\bm{y}^{(t)}$ are binary random vectors, this is also equivalent to 
\begin{align}
\begin{split}
&g_t\lp\bm{m}^{(t)},\,\bm{y}^{(t)}; \ \bm{\xi}\rp= g_t\lp\bm{m}^{(t)},\,\bm{y}^{(t)}; \ \bm{\xi}^*\rp\ \ \mathrm{or}\ \ g_t\lp\bm{m}^{(t)},\,\bm{y}^{(t)}; \ \bm{\xi}^*\rp=0 \quad \mathrm{for}~\forall t,\,\forall\bm{y}^{(t)}\\
\Rightarrow\ &\bm{\xi}=\bm{\xi}^*. \label{2.6}
\end{split} 
\end{align}
\end{definition}

The likelihood  took into account of $\bm{M}$ is called full information maximum likelihood(FIML) and say $L_N(\bm{\xi})$, where $N$ is sample size. It becomes
\begin{align*}
L_N(\bm{\xi}):= \prod_{i=1}^N\sum_{t=1}^T\bm{1}_{\{\bm{m}_i=\bm{m}^{(t)}\}} g_t\lp\bm{m}^{(t)},\,\bm{y}^{(t)}_i\,;\,\bm{\xi}\rp.
\end{align*}
Let $L(\bm{\xi})$ be a function which is the destination of the log-likelihood $\frac{1}{N}\log(L_N(\bm{\xi}))$ as  $N$ tends to infinity. There exists such a function by the strong law of large numbers, 
\begin{align*}
\lim_{N\to\infty}\frac{1}{N}\log\{L_N(\bm{\xi})\}
&=\lim_{N\to\infty}\frac{1}{N}\sum_{i=1}^N \log \llp \sum_{t=1}^T \bm{1}_{\{\bm{m}_i=\bm{m}^{(t)}\}}g_t(\bm{m}^{(t)},\,\bm{y}^{(t)}_i;\,\bm{\xi})\rrp\\
&= E_{\bm{\xi}^*}\lllp\log \llp g_t(\bm{m}^{(t)},\,\bm{y}^{(t)};\,\bm{\xi})\rrp\rrrp\lp=:L(\bm{\xi})\rp \hspace{27mm} \mathrm{a.s.}~P_{\bm{\xi}^*}\\
&= \sum_{t=1}^T\sum_{\bm{y}\in\{0,\,1\}^{\otimes T}} \log \llp g_t(\bm{m}^{(t)},\,\bm{y}^{(t)};\,\bm{\xi})\rrp g\lp\bm{m}^{(t)},\,\bm{y}\ ; \ \bm{\xi}^*\rp \quad \mathrm{a.s.}~P_{\bm{\xi}^*}, 
\end{align*}
where $E_{\bm{\xi}^*}[\,\cdot\,]$ represents the expectation under the probability measure $P_{\bm{\xi}^*}$.
We can obtain an important equation \eqref{2.4}, which is needed to assure asymptotic properties, 
\begin{align}
\sup_{\bm{\xi}\in\Xi_{\varepsilon}} L(\bm{\xi})<L(\bm{\xi}^*)\qquad \mathrm{for~} \forall \varepsilon>0,
\label{2.4}
\end{align}
where $\bm{\Xi}_{\varepsilon}:=\{\bm{\xi}\in \Xi\mid |\bm{\xi}-\bm{\xi}^*|\geq \varepsilon\}$. For this proof, we need two more assumptions in addition to \eqref{2.3}: compactness of $\Xi$ and continuity of $L(\bm{\xi})$. This proof is given in Appendix \ref{A.2}. For these reasons, we focus on whether \eqref{2.6}, which is equivalent to \eqref{2.3}, holds or not in the following discussion.
\par
We assume a natural assumption that the distribution of $\bm{Y}$and the missing-data mechanism are AR($p$) ($p\geq 1$) model. By transforming $g_t(\bm{m}^{(t)},\,\bm{y}^{(t)};\ \bm{\xi})$ proficiently, we can obtain next lemma.

\begin{lemma}
\label{lem:2.2.1}
If the distribution of $\bm{Y}$and the missing-data mechanism are  AR($p$)$(p\geq 1)$ model,  the joint probability function of $(\bm{M},\,\bm{Y}^{(t)})\ (t=1,\,\ldots,\,T)$ becomes
\begin{align}
g_t(\bm{m}^{(t)},\,\bm{y}^{(t)}\mid \bm{\xi}) &=\pi(y_{1};~\bm{\theta}_1)\llp \prod_{s=2}^{t}f_s(\bm{h}^{(s/p)};~\bm{\theta}_t,\bm{\tau}_t)\rrp^{\bm{1}_{\{t\geq 2\}}}\nonumber\\
&\quad \times\llp 1- \bm{1}_{\{t\leq T-1\}}\sum_{y_{t+1}=0}^1f_{t+1}\lp[\bm{h}^{(t/(p-1))'},~y_{t+1}]';~\bm{\theta}_t,\bm{\tau}_t\rp\rrp ,
\label{2.7}
\end{align}
where 
\begin{align*}
\pi(y_{1};~\theta_1)&:= P(Y_1=y_1;~\theta_1),\\
f_t(\bm{h}^{(t/p)};~\bm{\theta}_t,\bm{\tau}_t)&:= P\lp M_{t}=0\mid M_{t-1}=0,~\bm{h}^{(t/p)};~\bm{\theta}_t,\bm{\tau}_t\rp P\lp Y_t=y_{t}\mid \bm{h}^{((t-1)/p)};~\bm{\theta}_t\rp,\\
&\hspace{90mm} \mathrm{for}~t=2,\ldots, T.
\end{align*}
\end{lemma}
The proof is given in Appendix \ref{A.3}. By this lemma, the likelihood $L_N(\bm{\xi})$ is represented by the function of $\pi(y_{1};~\theta_1)$ and $f_t(\bm{h}^{(t/p)};~\bm{\theta}_t,\bm{\tau}_t)~(t\geq 2)$. The next theorem follows easily from the that previous lemma.
\begin{theorem}
The condition  \eqref{2.6} holds true if and only if the following conditions are met: For $t\geq 2$,
\begin{align}
&\pi_1(y_1;\theta_1)= \pi_1(y_1;\theta^*_{1})\ \ \mathrm{or}\ \  \pi_1(y_1;\theta^*_{1})=0\quad \mathrm{for}\ \forall y_1\ \Rightarrow\ \theta_1 =\theta^*_{1}  ,
\label{2.10} \\
\begin{split}
&f_t(\bm{h}^{(t/p)};~\bm{\theta}_t,\bm{\tau}_t)= f_t(\bm{h}^{(t/p)};~\bm{\theta}^*_{t},\bm{\tau}^*_{t})\ \ \mathrm{or}\ \  f_t(\bm{h}^{(t/p)};~\bm{\theta}^*_{t},\bm{\tau}^*_{t})=0\quad \mathrm{for}\ \forall \bm{h}^{(t/p)}\\
\Rightarrow\ &(\bm{\theta}_t,\bm{\tau}_t)=(\bm{\theta}^*_{t},\bm{\tau}^*_{t}).
\end{split}
\label{2.11}
\end{align}
\end{theorem}  

All we have to do is to check condition \eqref{2.10} and \eqref{2.11}. The condition \eqref{2.10} is obvious seen from the definition \eqref{pi} and hence we consider only \eqref{2.11}. For example, if $p=1,~t=2$, 
\begin{align*}
f_2(\bm{h}^{(2/1)};~\bm{\theta}_2,\bm{\tau}_2)=P\lp M_{2}=0\mid M_{1}=0,~y_1,y_2;~\bm{\tau}_2\rp P\lp Y_t=y_{2}\mid y_1;~\bm{\theta}_2\rp.
\end{align*}
For $p=1,~t=2$, \eqref{2.11} is equivalent to the following condition:
\begin{align}
&P\lp M_{2}=0\mid M_{1}=0,\,i,\,j\,;\,\bm{\tau}_2\rp P\lp Y_2=j\mid i\,;~\bm{\theta}_2\rp \label{2.12}\\
&= P\lp M_{2}=0\mid M_{1}=0,\,i,\,j\,;\,\bm{\tau}^*_{2}\rp P\lp Y_2=j\mid i\,;~\bm{\theta}^*_{2}\rp\quad \mathrm{for}~\forall i,\,j=0,\,1 \nonumber\\
\Rightarrow& (\bm{\theta}_2,\bm{\tau}_2)= (\bm{\theta}^*_{2},\bm{\tau}^*_{2})\nonumber
\end{align}
There are 4 constraints in \eqref{2.12} and if all of them were linear equations, the number of parameters had to be smaller than or equal 4. Since we have 5 parameters in the Logistic AR(1) model, this would imply that the model did not have identifiability. In binary data analysis, they are usually non-linear expressions, but it is worth verifying this condition; that is to say in general whether  dim($\bm{\xi}_t$), representing the number of parameters used at time $t$, is smaller than or equal to $2^{\mathrm{dim}(\bm{h}^{(t/p)})}$, representing the number of constraints. Note generally that, we have to check all the above expressions. 

\subsection{Identifiability of Logistic AR(1) model}
In the logistic AR(1) model defined in \eqref{Logistic}-\eqref{pi}, there are 5 parameters at each time $t~(t\geq 2)$. From the previous discussion, it seems that logistic AR(1) model dose not have identiability, and  in fact, it does not. 
\begin{proposition}
\label{prop:2.2.3}
The logistic AR(1) model defined in \eqref{Logistic}-\eqref{pi} does not have identifiability. 
\end{proposition}
\begin{proof}
For simplicity we write
\begin{align*}
(a_{20},a_{21},b_{20},b_{21},b_{22})&:=(\exp(\theta_{20}),\exp(\theta_{21}),\exp(-\tau_{20}),\exp(-\tau_{21}),\exp(-\tau_{22}))\\
(a^*_{20},a^*_{21},b^*_{20},b^*_{21},b^*_{22})&:=(\exp(\theta^*_{20}),\exp(\theta^*_{21}),\exp(-\tau^*_{20}),\exp(-\tau^*_{21}),\exp(-\tau^*_{22}))
\end{align*}
and prove it only for $t=2$, i.e., for $(a^*_{20},a^*_{21},b^*_{20},b^*_{21},b^*_{22})$ there exists $(a_{20},a_{21},$ $b_{20},b_{21},b_{22})\neq (a^*_{20},a^*_{21},b^*_{20},b^*_{21},b^*_{22})$ such that \eqref{2.12} holds. To show the result, we fix $a_{20}$ some value(say, $\tilde{a}_{20}$) which is not $a^*_{20}$ and prove that \eqref{2.12} holds if and only if the rest parameters $(a_{21},b_{20},b_{21},b_{22})$ is written by a function of ($\tilde{a}_{20},a^*_{20},$ $a^*_{21},b^*_{20},b^*_{21},b^*_{22}$), which shows that the logistic AR(1) model is not identified.

In the logistic AR(1) model, \eqref{2.12} is represented as
\begin{numcases}
  {}
 (1+a_{20})(1+b_{20})=(1+a^*_{20})(1+b^*_{20}) & \label{2.13} \\
 \lp1+\frac{1}{a_{20}}\rp(1+b_{20}b_{22})=\lp1+\frac{1}{a^*_{20}}\rp(1+b^*_{20}b^*_{22}) &\label{2.14} \\
 (1+a_{20}a_{21})(1+b_{20}b_{21})=(1+a^*_{20}a^*_{21})(1+b^*_{20}b^*_{21}) &\label{2.15}\\
\lp1+\frac{1}{a_{20}a_{21}}\rp(1+b_{20}b_{21}b_{22})=\lp1+\frac{1}{a^*_{20}a^*_{21}}\rp(1+b^*_{20}b^*_{21}b^*_{22})&\label{2.16}
\end{numcases}
where $a_{20}=\tilde{a}_{20}$ and all the parameters $a_{20},\,a_{21},\,b_{20},\,b_{21},\,b_{22},\,a_{20}^*,\,a^*_{21},\,b^*_{20},\,b^*_{21},\,b^*_{22}$ are positive. 
By \eqref{2.13}, we have
\begin{align}
b_{20}=\frac{(1+a^*_{20})(1+b^*_{20})}{1+\tilde{a}_{20}}-1
\label{2.17}
\end{align}
and by \eqref{2.14},
\begin{align}
b_{22}=\frac{\tilde{a}_{20}(1+a^*_{20})(1+b^*_{20}b^*_{22})-a^*_{22}(1+\tilde{a}_{20})}{a^*_{20}\{(1+a^*_{20})(1+b^*_{20})-(1+\tilde{a}_{20})\}}.
\label{2.18}
\end{align}
To guarantee that $b_{20},b_{22}>0$, $\tilde{a}_{20}$ must satisfy
\begin{align}
\frac{a^*_{20}}{1+b^*_{20}b^*_{22}(1+a^*_{20})}<\tilde{a}_{20}<(1+a^*_{20})(1+b^*_{20})-1.
\label{2.19}
\end{align}
We assume this condition for $\tilde{a}_{20}$. By multiplying $\tilde{a}_{20}a_{21}b^*_{20}b^*_{21}$ \eqref{2.16} and dividing by $\eqref{2.17}$, we obtain
\begin{align}
a_{21}=\frac{a^*_{20}a^*_{21}(1+b^*_{20}b^*_{21})(1+b_{20}b_{21}b_{22})}{(1+b^*_{20}b^*_{21}b^*_{22})(1+b_{20}b_{21})}.
\label{2.20}
\end{align}
By replacing $b_{20},b_{22}$ with \eqref{2.17}, \eqref{2.18}, we have
\begin{align}
b_{21}=\frac{b^*_{20}b^*_{21}\llp a^*_{20}a^*_{21}b^*_{22}(b^*_{20}b^*_{21}+1)+b^*_{20}b^*_{21}b^*_{22}+1 \rrp}
{\lp \frac{(a^*_{20}+1)(b^*_{20}+1)}{\tilde{a}_{20}+1}-1 \rp 
\lp \frac{a^*_{20}a^*_{21}(b^*_{20}b^*_{21}+1)\lllp \tilde{a}_{20}\llp (a^*_{20}+1)b^*_{20}b^*_{22}+1 \rrp -a^*_{20}\rrrp}{a^*_{20}(a^*_{20}+1)(b^*_{20}+1)-(\tilde{a}_{20}+1)} \rp}.
\label{2.21}
\end{align}
Furthermore, by substituting \eqref{2.17}, \eqref{2.18}, and \eqref{2.20} for \eqref{2.21}, we can obtain $a_{21}$ by a function of 
($\tilde{a}_{20},a^*_{20},a^*_{21},b^*_{20},b^*_{21},b^*_{22}$). Therefore,  all the parameters $a_{20},a_{21},b_{20},b_{21},b_{22}$ are represented by them.
\end{proof}
\subsection{Identifiability of Logistic AR(2) model}
Usually, if smaller models e.g., AR(1) model, do not have identifiability, neither do larger models e.g., AR(2) model. The general theory is not always true, however, since poor information at time $t$ makes the Logistic AR(1) model unidentified, use of past information may make the Logistic AR(2) model identified. Recall that the AR(1) model defined in \eqref{Logistic}-\eqref{pi} is not identified because the model imposes  $2^2=4$ constraints  in \eqref{2.12} with the larger number 5 of the parameters at each time $t~(t\geq 2)$  in the model;  see the graphical model in Figure \ref{fig:2.1}. Here the parameters enclosed by a broken line denote intercepts in the model and the others parameters denote each coefficient of the nearest arrow in the model. As we have already developed the expression for the joint distribution function of  AR($p$) model in Lemma \ref{lem:2.2.1}, if $t=3,p=2$, one parameter is added about $\bm{Y}$'s serial correlation: we have 6 parameters against $2^3=8(>6)$ constraints. It is seen visually in Figure \ref{fig:2.2}. Thus, for $t\geq 3$, we can claim that \eqref{2.11}  holds as follows. 
\begin{figure}[htbp]
 \begin{minipage}{0.45\hsize}
  \begin{center}
  \includegraphics[scale=0.5]{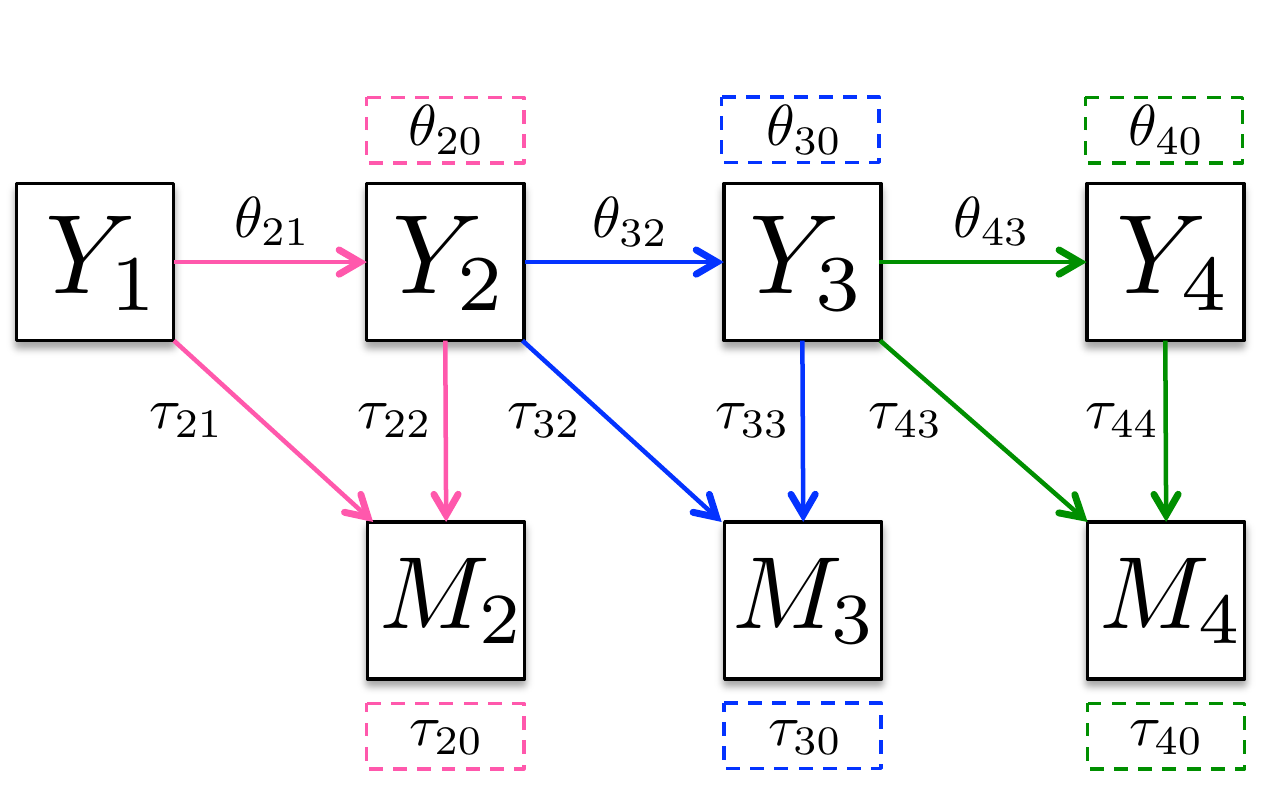}
  \end{center}
  \caption{AR(1) model}
  \label{fig:2.1}
 \end{minipage}
 \begin{minipage}{0.45\hsize}
  \begin{center}
  \includegraphics[scale=0.5]{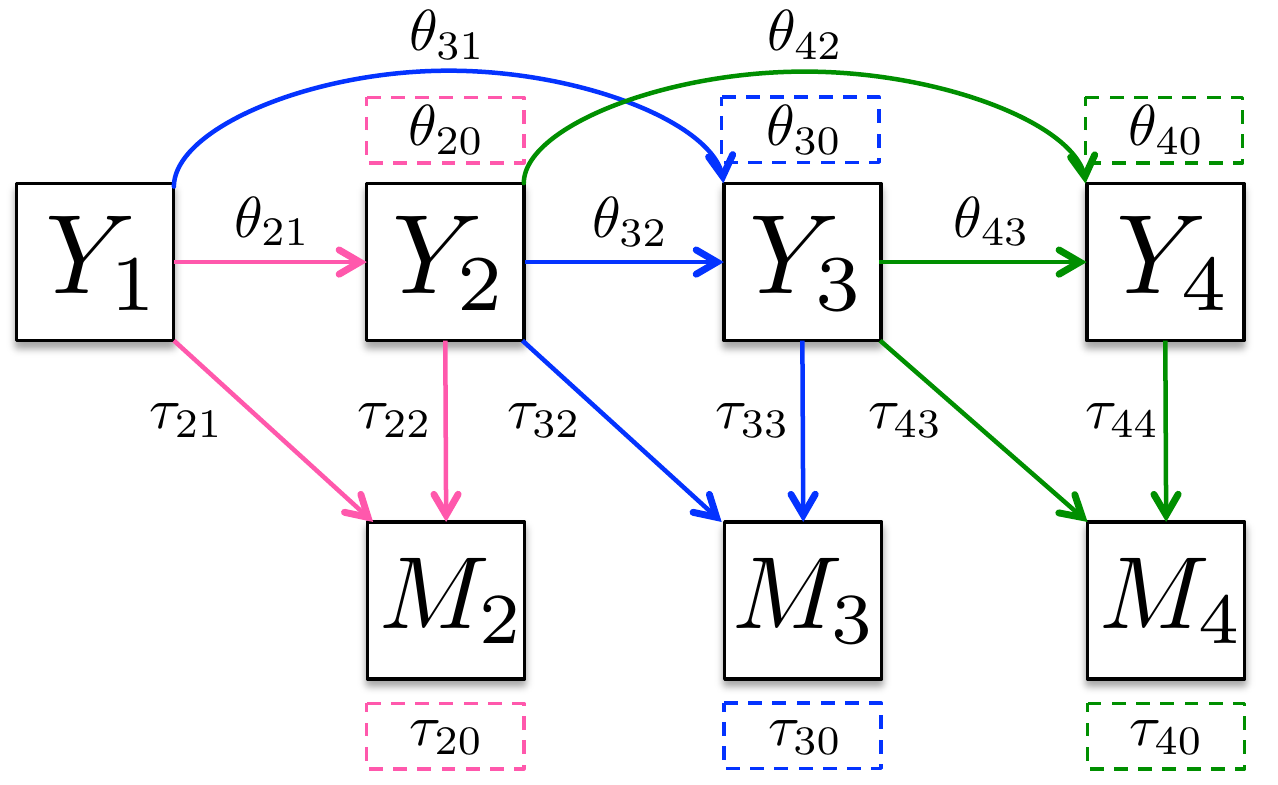}
  \end{center}
  \caption{AR(2) model}
  \label{fig:2.2}
 \end{minipage}
\end{figure}
\begin{proposition}
\label{prop:2.2.4}
In logistic AR(2) model, for $t\geq 3$, 
\begin{align*}
f_t(\bm{h}^{(t/2)};~\bm{\theta}_t,\bm{\tau}_t) = f_t(\bm{h}^{(t/2)};~\bm{\theta}^*_{t},\bm{\tau}^*_{t})\ \ \mathrm{for}\ \forall \bm{h}^{(t/2)} \ &\Rightarrow\ 
(\bm{\theta}_t,\bm{\tau}_t) = (\bm{\theta}^*_{t},\bm{\tau}^*_{t}) 
\end{align*}
holds if and only if  $\theta_{tt-2}$ is not zero.
\end{proposition}
\begin{proof}
The result what we want obtain is, for every $(y_{t-2},y_{t-1},y_t)' \in \{0,1\}^{\otimes 3}$,
\begin{align}
&\frac{1}{1+\exp\{-(\tau_{t0}+\tau_{tt-1}y_{t-1}+\tau_{33}y_{t})\}}\frac{1}{1+\exp\{(-1)^{y_{t}}(\theta_{t0}+\theta_{tt-2}y_{t-2}+\theta_{tt-1}y_{t-1})\}}\nonumber\\
&=\frac{1}{1+\exp\{-(\tau^*_{t0}+\tau^*_{tt-1}y_{t-1}+\tau^*_{33}y_{t})\}}\frac{1}{1+\exp\{(-1)^{y_{t}}(\theta^*_{t0}+\theta^*_{tt-2}y_{t-2}+\theta^*_{tt-1}y_{t-1})\}}
\label{2.22}\\
\Rightarrow &  (\bm{\theta}_3,\bm{\tau}_3)= (\bm{\theta}^*_{3},\bm{\tau}^*_{3}). \nonumber
\end{align} 
We prove only for $t=3$ since the above expression is same for $t\geq 3$. \eqref{2.22} is equivalent to following 8 expressions.
\begin{numcases}
  {}
 (1+a_{30})(1+b_{30})=(1+a^*_{30})(1+b^*_{30}) & \label{2.23} \\
 (1+a_{30}a_{31})(1+b_{30})=(1+a^*_{30}a^*_{31})(1+b^*_{30}) & \label{2.24} \\
 \lp1+\frac{1}{a_{30}}\rp(1+b_{30}b_{33})=\lp1+\frac{1}{a^*_{30}}\rp(1+b^*_{30}b^*_{33}) &\label{2.25} \\
  \lp1+\frac{1}{a_{30}a_{31}}\rp(1+b_{30}b_{33})=\lp1+\frac{1}{a^*_{30}a^*_{31}}\rp(1+b^*_{30}b^*_{33}) &\label{2.26} \\
 (1+a_{30}a_{32})(1+b_{30}b_{32})=(1+a^*_{30}a^*_{32})(1+b^*_{30}b^*_{32}) &\label{2.27}\\
 (1+a_{30}a_{31}a_{32})(1+b_{30}b_{32})=(1+a^*_{30}a^*_{31}a^*_{32})(1+b^*_{30}b^*_{32}) &\label{2.28}\\
\lp1+\frac{1}{a_{30}a_{32}}\rp(1+b_{30}b_{32}b_{33})=\lp1+\frac{1}{a^*_{30}a^*_{32}}\rp(1+b^*_{30}b^*_{32}b^*_{33})&\nonumber \\
\lp1+\frac{1}{a_{30}a_{31}a_{32}}\rp(1+b_{30}b_{32}b_{33})=\lp1+\frac{1}{a^*_{30}a^*_{31}a^*_{32}}\rp(1+b^*_{30}b^*_{32}b^*_{33})&\nonumber
\end{numcases}
where
\begin{align*}
(a_{30},a_{31},a_{32},b_{30},b_{32},b_{33})&:=(\exp(\theta_{30}),\exp(\theta_{31}),\exp(\theta_{32}),\exp(-\tau_{30}),\exp(-\tau_{32}),\exp(-\tau_{33}))\\
(a^*_{30},a^*_{31},a^*_{32},b^*_{30},b^*_{32},b^*_{33})&:=(\exp(\theta^*_{30}),\exp(\theta^*_{31}),\exp(\theta^*_{32}),\exp(-\tau^*_{30}),\exp(-\tau^*_{32}),\exp(-\tau^*_{33})).
\end{align*}
By dividing \eqref{2.23} by \eqref{2.24} and \eqref{2.26} by \eqref{2.25}, we have
\begin{align}
\frac{1+a_{30}}{1+a_{30}a_{31}}&=\frac{1+a^*_{30}}{1+a^*_{30}a^*_{31}},
\label{2.29} \\
\frac{a_{31}(1+a_{30})}{1+a_{30}a_{31}}&=\frac{a^*_{31}(1+a^*_{30})}{1+a^*_{30}a^*_{31}}.
\label{2.30}
\end{align}
By substituting \eqref{2.29} for \eqref{2.30}, we obtain $a_{31}=a^*_{31}$. Using this equation for \eqref{2.29} again,
\begin{align*}
(a^*_{31}-1)(a_{30}-a^*_{30})=0.
\end{align*}
Thus, if $a^*_{31}\neq 1$, $a_{30}=a^*_{30}$, otherwise $a_{30}=c$, where $c$ is an arbitrary positive constant. Hence, when $a^*_{31}= 1$, the model does not have identifiability. If $a^*_{31}\neq 1$, by dividing \eqref{2.27} by \eqref{2.28} and substituting $a_{30}=a^*_{30}$ and $a_{31}=a^*_{31}$ for it, we can obtain $(a^*_{31}-1)(a_{32}-a^*_{32})=0$, which implies $a_{32}=a^*_{32}$. Therefore,
\begin{align*}
a^*_{31}\neq 1 \Leftrightarrow  (a_{30},\,a_{31},\,a_{32}) = (a^*_{30},\,a^*_{31},\,a^*_{32})
\end{align*}
holds. From \eqref{2.23}, \eqref{2.25} and \eqref{2.27}, it is obviously seen that
\begin{align*}
(a_{30},\,a_{31},\,a_{32}) = (a^*_{30},\,a^*_{31},\,a^*_{32}) \Leftrightarrow  
(b_{30},\,b_{32},\,b_{33}) = (b^*_{30},\,b^*_{32},\,b^*_{33})
\end{align*}
holds. Hence, 
\begin{align*}
a^*_{31}\neq 1 \Leftrightarrow
(a_{30},\,a_{31},\,a_{32},\,b_{30},\,b_{32},\,b_{33}) = (a^*_{30},\,a^*_{31},\,a^*_{32},\,b^*_{30},\,b^*_{32},\,b^*_{33}) 
\end{align*}
Thus, we have the conclusion. 
\end{proof}
Evidently seen by two graphical models Figure \ref{fig:2.1}-\ref{fig:2.2}, AR(2) model does not have identifiability when $t=2$ as yet for the same reason with AR(1) model. We have to add additional information into the model. 
\subsection{Examples of Identifiable Models}
Consider when there is no missing data at time $t=2$. In this case, we can fix 2 parameters $\tau_{21}=\tau_{22}=0$. Hence, there are 3($=5-2$) parameters against 4($>$3) constraints. It may seems the model has identifiability, in fact, it does. This is not proved here, but can do similar way with Proposition \ref{prop:2.2.4}. In another case, consider we have one binary covariate which is invariant for all times and has no missing data such as the information of dose at \citet{machin88}. More specifically, reconsider model as
\begin{align*}
P(Y_1=1;\pi_1)&=\theta_1,\\
P(Y_t=1\mid y_{t-1},x\,;\,\theta_{t0},\theta_{tt-1},\beta_t)&=\mathrm{expit}(\theta_{t0}+y_{t-1}\theta_{tt-1}+\beta_t x)\qquad \mathrm{for}\quad t\geq 2.
\end{align*}
in AR(1) model where  $x$ is a covariate. Here, let $x$ be a binary random variable to consider the worst case to have identifiability.  This graphical model is shown in Figure \ref{fig:2.3}. Note that this models is a conditional model. At time $t\geq 2$, there are 6 parameters against 8 constraints for each time $t$ same as AR(2) model when $t\geq 3$. At time $t=1$, we can fix 2 parameters which has an effect to missing data indicator as 0 since there is no missing data on covariate, the model is identified for the same reason when there is no missing data at $t=2$. Thus, this model probably has identification, in fact, it does. This is not proved here, but can do similar way with Proposition \ref{prop:2.2.4}. Moreover, we can show AR(2) with one covariate model has also identifiability whose graphical model is shown in Figure \ref{fig:2.4}.
\begin{figure}[htbp]
 \begin{minipage}{0.43\hsize}
  \begin{center}
  \includegraphics[scale=0.5]{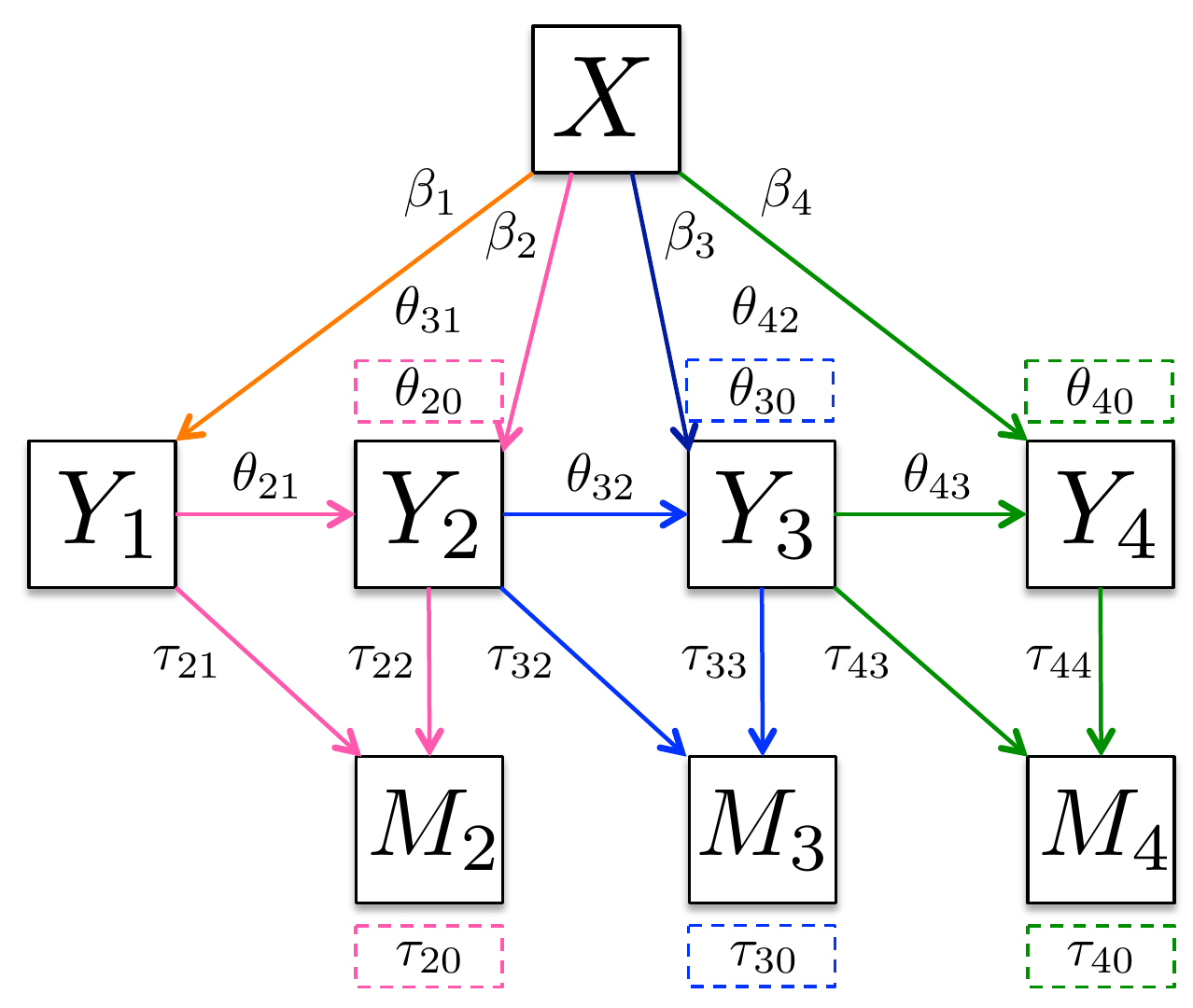}
  \end{center}
  \caption{AR(1) model with one covariate}
  \label{fig:2.3}
 \end{minipage}
  \begin{minipage}{0.8\hsize}
 \end{minipage}
 \begin{minipage}{0.43\hsize}
  \begin{center}
  \includegraphics[scale=0.5]{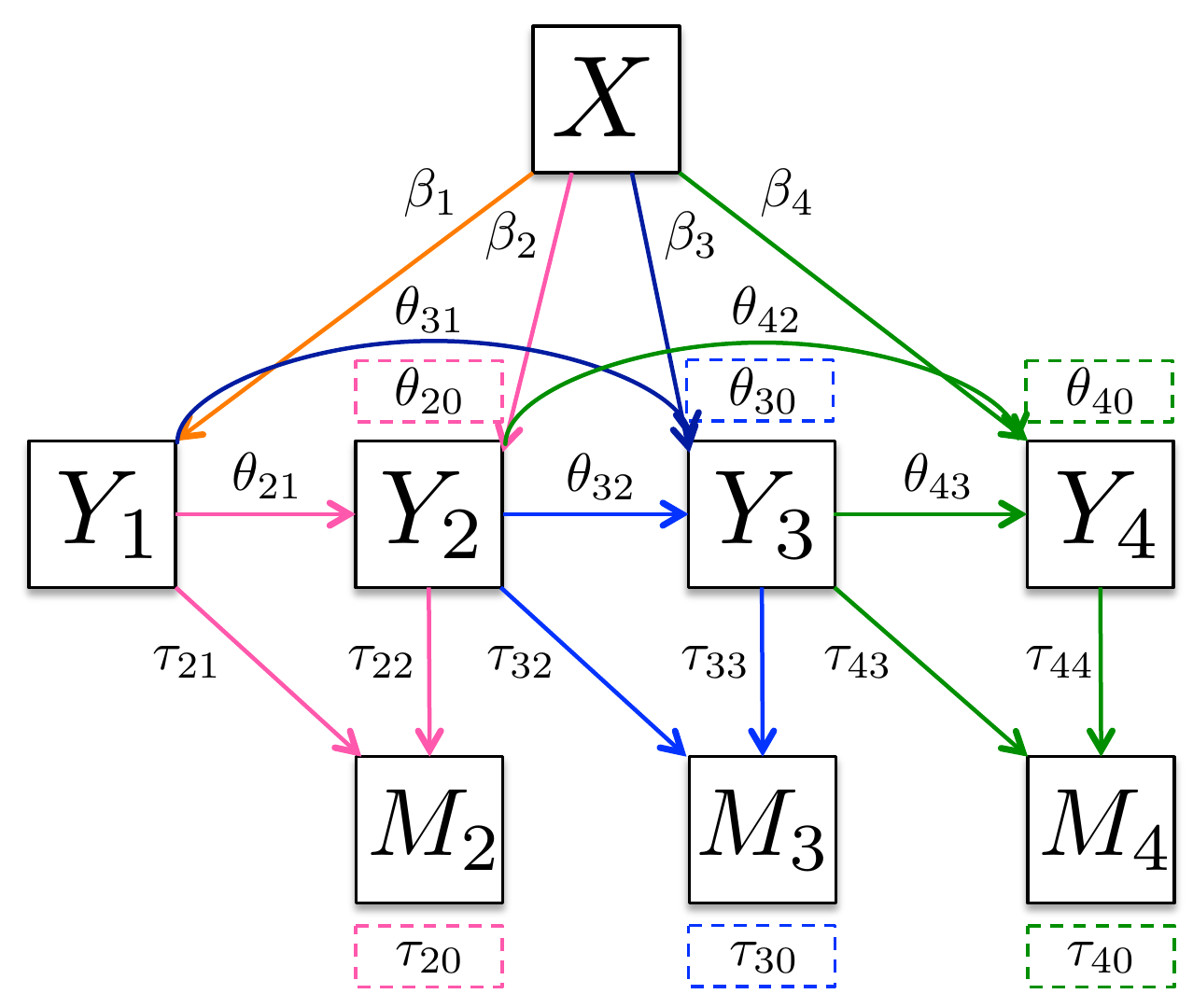}
  \end{center}
  \caption{AR(2) model with one covariate}
  \label{fig:2.4}
 \end{minipage}
\end{figure}
\section{Real Data Analysis}
In this section, we analyze the data given in Table \ref{tb:1.1} by the AR(2) model with one covariate to see effects of a contraceptive DMPA at time $t=1,\,\ldots,\,4$ where it is guaranteed the model has identifiability. In addition, we discuss its missing-data mechanism and select the best model by the likelihood-ratio test. 
\subsection{Parameter Estimation}
Parameters are estimated as MLE and FIML is defined by \eqref{1.2} and \eqref{2.7} as follows
\begin{align*}
L_{1151}(\bm{\xi}) &=\prod_{t=1}^4\prod_{i\in I_{t}}\pi(y_{1i};~\bm{\theta}_1)\llp \prod_{s=2}^{t}f_s(\bm{h}^{(s/2)}_i;~\bm{\theta}_t,\bm{\tau}_t)\rrp^{\bm{1}_{\{t\geq 2\}}}\nonumber\\
&\quad \times\llp 1- \bm{1}_{\{t\leq T-1\}}\sum_{y_{t+1}=0}^1f_{t+1}\lp[\bm{h}^{(t/1)'}_i,~y_{t+1}]';~\bm{\theta}_t,\bm{\tau}_t\rp\rrp,
\end{align*}
where the sample size is $N=1151$.
MLE of $\bm{\xi}_1,\,\ldots,\,\bm{\xi}_4$ can be calculated separately since each parameter is separated in the above likelihood. We use a ``optim'' function to optimize $L_N(\bm{\xi})$ with a programing language R. 

Looking at the ``s.d. (standard deviation)'' term in Table \ref{tb:2.2}, we can see this model has identifiability indeed; if this model does not have identifiability, all of s.d. values diverge to infinity. First, the result of $\beta_t$ means dose of 150mg has a significant difference from zero to the effect of contraception compared to that of 100mg for 6 months and 9 months ($t=2,3$) in terms of $p$-value; this is the direct effect, but we can see the same fact with respect to the total effect\citep[see][]{matsuyama04}. Next, the result of serial correlation is intensively positive. This means once DMPA takes effect, it is also liable to do next time independently from dose. Finally, $\bm{\tau}_2,\,\bm{\tau}_3,\,\bm{\tau}_4$, which are parameters on the missing-data mechanism,  are the most difficult to interpret since the values of s.d. are relatively larger than those of the others parameters. Due to this fact, all parameters do not have significantly difference from zero. In the next subsection, we consider best model by constructing sub-models from this full-model and choose by likelihood ratio test.

\begin{table}[htbp]
\begin{center}
\caption{Results of Parameter Estimation by AR(2) model}
\begin{tabular}{crrr}
\hline
parameter & MLE & s.d. & p-value  \\ \hline 
$ \beta_1$ & 0.124 & 0.149 & 0.406\\ 
$ \beta_2$ & 0.390 & 0.152 & 0.010\\ 
$ \beta_3$ & 0.440 & 0.160 & 0.006\\ 
$ \beta_4$ & 0.124 & 0.149 & 0.406\\ \hline
$\theta_{21}$ & 1.851 & 0.215 & 0.000\\ 
$\theta_{32}$ & 2.014 & 0.195 & 0.000\\ 
$\theta_{43}$ & 1.794 & 0.228 & 0.000\\ 
$\theta_{31}$ & 0.852 & 0.235 & 0.000\\ 
$\theta_{42}$ & 1.382 & 0.233 & 0.000\\ \hline
$\tau_{21}$ & -0.506 & 0.680 & 0.457\\ 
$\tau_{32}$ & -0.276 & 0.546 & 0.613\\ 
$\tau_{43}$ & -1.067 & 0.506 & 0.035\\ 
$\tau_{22}$ & -0.079 & 1.544 & 0.959\\ 
$\tau_{33}$ & -0.719 & 1.231 & 0.559\\ 
$\tau_{44}$ & 0.939 & 0.931 & 0.313\\ \hline
\end{tabular}
\label{tb:2.2}
\end{center}
\end{table}
\subsection{Model Selection}
If there is data involves missing values, information criterion such as AIC and BIC can not be used. Then,  we choose best model by heuristic way: likelihood ratio test. First, we test  missing-data mechanism ``MCAR v.s. NMAR'' and ``MAR v.s. NMAR'' as with \citet{diggle94}. Denote MLE under a constraint $\tau_{21}=\tau_{22}=\tau_{32}=\tau_{33}=\tau_{43}=\tau_{44}=0$ by $\hat{\bm{\xi}}_{MCAR}$  , under a constraint $\tau_{22}=\tau_{33}=\tau_{44}=0$ by $\hat{\bm{\xi}}_{MAR}$  and no constraints by $\hat{\bm{\xi}}_{NMAR}$, i.e., full-model's MLE. We set the probability of type I error to 0.05 in following two tests of its missing-data mechanism. In this settings, the deviance of MCAR and NMAR is,
\begin{align*}
-2\log \frac{L_{1151}(\hat{\bm{\xi}}_{MCAR})}{L_{1151}(\hat{\bm{\xi}}_{NMAR})} =27.157 > \chi^2_{6}(0.05)=12.592.
\end{align*}
and the deviance of MAR and NMAR is,
\begin{align*}
-2\log \frac{L_{1151}(\hat{\bm{\xi}}_{MAR})}{L_{1151}(\hat{\bm{\xi}}_{NMAR})} =1.204 < \chi^2_{3}(0.05)=7.814.
\end{align*}
Hence, the missing-data mechanism is not MCAR, but not to say NMAR. Then, consider sub-models in which using less than 3 parameters from 6 parameters($\tau_{21},\,\tau_{22},\,\tau_{32},\,\tau_{33},\,\tau_{43},\,\tau_{44}$): ${}_6 C_1+{}_6 C_2+{}_6 C_3=41$ sub-models.  The result is shown In Table \ref{tb:2.3}, red letters stands for parameter sets whose deviance is smallest when the estimated number of parameters is $1$(No. 4), $2$(No. 9) or $3$(No. 29) and green letters stands for the parameter sets using at least one parameter at one time. This results show that when estimated number of parameters is 3, the deviance declines drastically compared to when that is $1$ or $2$.  Needless to say, the more number of estimated parameters is, the smaller its deviance becomes.  However, it leads to the number of parameters more than necessary.  The deviance whose number is 27, 29, 33, and 35  is relatively smaller than others' deviance and they are not different significantly. Therefore, we asserts that these 4 models are best model, which are all green letters, namely, chosen one parameter by one time, where No. 27 is a MAR mechanism, but we can not choose which one is best from this data. 
\begin{table}[htbp]
\begin{center}
\caption{Sub-models' deviation}
\begin{tabular}{ccr|ccr|ccr}
\hline

No. &\shortstack{estimated\\ parameter} & \shortstack{deviance\\ ~} &
No. &\shortstack{estimated\\ parameter} & \shortstack{deviance\\ ~} & 
No. &\shortstack{estimated\\ parameter} & \shortstack{deviance\\ ~}\\ \hline 
1& $\tau_{21}$& 18.670 & 15 &$\tau_{22},\tau_{44}$& 14.912 & 29&$\tau_{21},\tau_{33},\tau_{43}$ & \textcolor[rgb]{1,0,0}{\textcolor[rgb]{1,0,0}{1.130}}\\ 
2& $\tau_{22}$& 19.268 &16&$\tau_{32},\tau_{33}$  & 15.878 & 30&$\tau_{21},\tau_{33},\tau_{44}$  & \textcolor[rgb]{0,0.6,0}{3.316}\\ 
3&$\tau_{32}$ & 16.242 &17&$\tau_{32},\tau_{43}$  & 9.692 &31& $\tau_{21},\tau_{43},\tau_{44}$ & 11.282\\ 
4&$\tau_{33}$ & \textcolor[rgb]{1,0,0}{16.160}&18 &$\tau_{32},\tau_{44}$  & 11.878 & 32& $\tau_{22},\tau_{32},\tau_{33}$ & 7.988\\ 
5&$\tau_{43}$ & 20.616 &19&$\tau_{33},\tau_{43}$  & 9.618 &33& $\tau_{22},\tau_{32},\tau_{43}$ & \textcolor[rgb]{0,0.6,0}{1.804}\\ 
6&$\tau_{44}$ & 22.802 &20&$\tau_{33},\tau_{44}$  & 11.804 &34&$\tau_{22},\tau_{32},\tau_{44}$  & \textcolor[rgb]{0,0.6,0}{3.990}\\ \cline{1-3}
7&$\tau_{21},\tau_{22}$& 18.666 &21& $\tau_{43},\tau_{44}$ & 19.770 &35& $\tau_{22},\tau_{33},\tau_{43}$ & \textcolor[rgb]{0,0.6,0}{1.728}\\ \cline{4-6}
8&$\tau_{21},\tau_{32}$ & 7.748 &22& $\tau_{21},\tau_{22},\tau_{32}$ & 7.744 & 36&$\tau_{22},\tau_{33},\tau_{44}$ & \textcolor[rgb]{0,0.6,0}{3.916}\\ 
9&$\tau_{21},\tau_{33}$ & \textcolor[rgb]{1,0,0}{7.672}&23 & $\tau_{21},\tau_{22},\tau_{33}$ & 7.670 &37&$\tau_{22},\tau_{33},\tau_{44}$  & 11.880\\ 
10&$\tau_{21},\tau_{43}$ & 12.128 &24& $\tau_{21},\tau_{22},\tau_{43}$ & 12.124 &38&$\tau_{22},\tau_{43},\tau_{44}$  & 9.336\\ 
11&$\tau_{21},\tau_{44}$ & 14.314 &25&$\tau_{21},\tau_{22},\tau_{44}$  & 14.290 &39&$\tau_{32},\tau_{33},\tau_{43}$  & 11.522\\ 
12&$\tau_{22},\tau_{32}$ & 8.346 &26& $\tau_{21},\tau_{32},\tau_{33}$ & 7.390 &40&$\tau_{32},\tau_{33},\tau_{44}$  & 8.846\\ 
13&$\tau_{22},\tau_{33}$ & 8.272 &27& $\tau_{21},\tau_{32},\tau_{43}$ & \textcolor[rgb]{0,0.6,0}{1.204} &41&$\tau_{32},\tau_{43},\tau_{44}$  & 8.772\\ 
14&$\tau_{22},\tau_{43}$ & 12.726 &28& $\tau_{21},\tau_{32},\tau_{44}$ & \textcolor[rgb]{0,0.6,0}{3.392} & & & \\ \hline
\end{tabular}
\label{tb:2.3}
\end{center}
\end{table}
\section{Conclusions and Discussion}
It is well known the identifiability of parameters often becomes problem because of poor information of the data $\bm{Y}$ in the analysis of binary data. In addition, if $\bm{Y}$ has missing values, the analysis which ignores missing data such as list wise deletion may make severe bias to the estimations such as mean and variance. Thus, the information of missingness must be taken into the model, however, which needs additional parameters according to its missing-data mechanism.  In particular, parameters prescribing whether the missing-data mechanism is NMAR or not tend to become unidentifiable. 

In this paper, we defined AR($p$) model which depends on the history only through the previous $p$ responses. Then, we gave a necessary and sufficient condition which makes its verification easy in AR($p$) model. For example, it is easily proved from the derived condition that even a simple AR(1) model does not have identifiability, but additional information makes it identifiable such as covariates or the fact data are not missing at the started two waves in a row.

However, this results are yielded under two critical assumptions; One is an assumption that there are no parameters which satisfies equality constraints  such as $\tau_{22}=\ldots=\tau_{TT}$ and the other is limiting the model to conditional models. First, if a model does not have identifiability, putting equality constraints on parameters is a natural idea. We have to rethink a condition to have identifiability under this constraints. Secondly, many complicated models such as marginal models and hybrid models are proposed by several authors and they are more used than conditional ones. Conditional models have some connection to these models since both of them factor same probability, but the likelihood of these models would become more complicated. We also need to derive the conditions correspond to these complicated models. 

\appendix
\section{Proof of \eqref{2.4}}
\label{A.2}
In this section, we prove \eqref{2.4} $\sup_{\bm{\xi}\in\Xi_{\varepsilon}} L(\bm{\xi})<L(\bm{\xi}^*)$, 
 under three assumptions: \eqref{2.3} holds,  compactness of a parameter space $\bm{\Xi}$ and continuity of $L(\bm{\xi})$. First, we show that 
\begin{align}
 L(\bm{\xi})<L(\bm{\xi}^*) \quad \mathrm{for}~\bm{\xi}\neq \bm{\xi}^*
 \label{a.2.1}
\end{align}
holds. In fact,
\begin{align*}
 &L(\bm{\xi})-L(\bm{\xi}^*)\\
 &=\sum_{t=1}^T\sum_{\bm{y}\in\{0,\,1\}^{\otimes T}} \log \llp g_t(\bm{m}^{(t)},\,\bm{y}^{(t)};\,\bm{\xi})\rrp g\lp\bm{m}^{(t)},\,\bm{y}\ ; \ \bm{\xi}^*\rp \\
&\quad -\sum_{t=1}^T\sum_{\bm{y}\in\{0,\,1\}^{\otimes T}} \log \llp g_t(\bm{m}^{(t)},\,\bm{y}^{(t)};\,\bm{\xi}^*)\rrp g\lp\bm{m}^{(t)},\,\bm{y}\ ; \ \bm{\xi}^*\rp\\
&=\sum_{t=1}^T\sum_{\bm{y}\in\{0,\,1\}^{\otimes T}} \log \frac{ g_t(\bm{m}^{(t)},\,\bm{y}^{(t)};\,\bm{\xi})}{g_t(\bm{m}^{(t)},\,\bm{y}^{(t)};\,\bm{\xi}^*)} g\lp\bm{m}^{(t)},\,\bm{y}\ ; \ \bm{\xi}^*\rp\\
&\leq\sum_{t=1}^T\sum_{\bm{y}\in\{0,\,1\}^{\otimes T}} \llp\frac{ g_t(\bm{m}^{(t)},\,\bm{y}^{(t)};\,\bm{\xi})}{g_t(\bm{m}^{(t)},\,\bm{y}^{(t)};\,\bm{\xi}^*)} -1\rrp g\lp\bm{m}^{(t)},\,\bm{y}\ ; \ \bm{\xi}^*\rp\\
&=\sum_{t=1}^T\sum_{\bm{y}^{(t)}\in\{0,\,1\}^{\otimes t}} g_t(\bm{m}^{(t)},\,\bm{y}^{(t)};\,\bm{\xi}) - \sum_{t=1}^T\sum_{\bm{y}\in\{0,\,1\}^{\otimes T}}g(\bm{m}^{(t)},\,\bm{y};\,\bm{\xi}^*) \\
&= 1-1 =0
\end{align*}
holds, where we have equality if and only if  \eqref{2.6} holds. Recall that the condition \eqref{2.6} implies we obtain equality if and only if $\bm{\xi}=\bm{\xi}^*$, hence this means \eqref{a.2.1}.

Then, we prove \eqref{2.4}. There is a minor gap between the condition \eqref{a.2.1} and \eqref{2.4}. To fill the gap, we have to show that there are no sequences $\{\bm{\xi}_n\}_{n\in\mathbb{N}}$ that tends to $\bm{\xi}_0\in\Xi \cap\Xi_{\varepsilon}$ such that attains $L(\bm{\xi}_0)=L(\bm{\xi}^*)$.
Suppose that there exists a sequence $\{\bm{\xi}_n\}_{n\in\mathbb{N}}\in \Xi \cap \Xi_{\varepsilon}$ such that $L(\bm{\xi}_n)\to L(\bm{\xi}^*)$. Due to the compactness of $\Xi \cap \Xi_{\varepsilon}$, there exists a subsequence $\{\bm{\xi}_{n_k}\}_{k\in \mathbb{N}}$ of $\{\bm{\xi}_n\}$ and $\bm{\xi}_0\in\Xi \cap \Xi_{\varepsilon}$ such that $\bm{\xi}_{n_k} \to \bm{\xi}_0$. By the continuity of $L$, $L(\bm{\xi}_{n_k})\to L(\bm{\xi}_0)=L(\bm{\xi}^*)$, which contradicts \eqref{a.2.1}. Thus, we have the conclusion. \qed
\section{Proof of Lemma \ref{lem:2.2.1}}
\label{A.3}
In this section, we prove Lemma \ref{lem:2.2.1}.
\begin{proof}
For simplicity, we abbreviate the parameter $\bm{\xi}$ in the following proof. \\
For $t=1,\ldots,T-1$,
\begin{align*}
&P\lp\bm{M}=\bm{m}^{(t)},~\bm{Y}^{(t)}=\bm{y}^{(t)}\rp\\
& =P\lp M_{t+1}=1,~M_{t}=0,~\bm{Y}^{(t)}=\bm{y}^{(t)}\rp\\
&=P\lp M_{t+1}=1\mid M_{t}=0,~\bm{y}^{(t)}\rp P\lp M_{t}=0,~\bm{Y}^{(t)}=\bm{y}^{(t)}\rp \\
&=\llp 1-\sum_{y_{t+1}=0}^1\frac{P\lp M_{t+1}=0,~\bm{Y}^{(t+1)}=[\bm{y}^{(t)'},y_{t+1}]'\rp}{P\lp M_{t}=0,~\bm{Y}^{(t)}=\bm{y}^{(t)}\rp} \rrp  P\lp M_{t}=0,~\bm{Y}^{(t)}=\bm{y}^{(t)}\rp \\
&= P\lp M_{t}=0,~\bm{Y}^{(t)}=\bm{y}^{(t)}\rp - \sum_{y_{t+1}=0}^1P\lp M_{t+1}=0,~\bm{Y}^{(t+1)}=[\bm{y}^{(t)'},y_{t+1}]'\rp .
\end{align*}
By including the case $t=T$, 
\begin{align}
&P\lp\bm{M}=\bm{m}^{(t)},~\bm{Y}^{(t)}=\bm{y}^{(t)}\rp\nonumber\\
&= P\lp M_{t}=0,~\bm{Y}^{(t)}=\bm{y}^{(t)}\rp - \bm{1}_{\{t\leq T-1\}}\sum_{y_{t+1}=0}^1P\lp M_{t+1}=0,~\bm{Y}^{(t+1)}=(\bm{y}^{(t)'},y_{t+1})'\rp . \label{a.3.1}
\end{align}
Due to the property of AR($p$) model, on the other hand, for $t=2,\ldots T$,
\begin{align*}
&P\lp M_{t}=0,~\bm{Y}^{(t)}=\bm{y}^{(t)}\rp\\
&=P\lp M_{t}=0\mid \bm{y}^{(t)}\rp P\lp\bm{Y}^{(t)}=\bm{y}^{(t)}\rp\\
&=P\lp M_{t}=0\mid M_{t-1}=0,~\bm{y}^{(t)}\rp P\lp M_{t-1}=0\mid \bm{y}^{(t)}\rp P\lp Y_t=y_{t}\mid \bm{y}^{(t-1)}\rp P\lp\bm{Y}^{(t-1)}=\bm{y}^{(t-1)}\rp  \\
&=P\lp M_{t}=0\mid M_{t-1}=0,~\bm{y}^{(t)}\rp P\lp M_{t-1}=0\mid \bm{y}^{(t-1)}\rp \\
&\quad\times P\lp Y_t=y_{t}\mid \bm{h}^{((t-1)/p)}\rp P\lp\bm{Y}^{(t-1)}=\bm{y}^{(t-1)}\rp  \\
&=\llp P\lp M_{t}=0\mid M_{t-1}=0,~\bm{h}^{(t/p)}\rp P\lp Y_t=y_{t}\mid \bm{h}^{((t-1)/p)}\rp\rrp \\
&\quad \times P\lp M_{t-1}=0\mid \bm{y}^{(t-1)}\rp P\lp\bm{Y}^{(t-1)}=\bm{y}^{(t-1)}\rp\\
&=\llp\prod_{s=2}^{t}P\lp M_{s}=0\mid M_{s-1}=0,~\bm{h}^{(s/p)}\rp P\lp Y_s=y_{s}\mid \bm{h}^{((s-1)/p)}\rp\rrp 
P(M_{1}=0\mid y_{1})P(Y_{1}=y_{1}) \\
&=P(Y_{1}=y_{1})\prod_{s=2}^{t}P\lp M_{s}=0\mid M_{s-1}=0,~\bm{h}^{(s/p)}\rp P\lp Y_s=y_{s}\mid \bm{h}^{((s-1)/p)}\rp \\
&=\pi(y_{1})\prod_{s=2}^{t}f_s(\bm{h}^{(s/p)}).
\end{align*}
By including the case $t=1$, 
\begin{align}
&P\lp M_{t}=0,~\bm{Y}^{(t)}=\bm{y}^{(t)}\rp=\pi(y_{1})\llp \prod_{s=2}^{t}f_s(\bm{h}^{(s/p)})\rrp^{\bm{1}_{\{t\geq 2\}}}
\label{a.3.2}
\end{align}
Hence, by substituting \eqref{a.3.2} for \eqref{a.3.1}, we have
\begin{align*}
g(\bm{m}^{(t)},\,\bm{y}^{(t)}) &= P\lp\bm{M}=\bm{m}^{(t)},~\bm{Y}^{(t)}=\bm{y}^{(t)}\rp\\
 &=\pi(y_{1})\llp \prod_{s=2}^{t} f_s(\bm{h}^{(s/p)})\rrp^{\bm{1}_{\{t\geq 2\}}}\llp 1- \bm{1}_{\{t\leq T-1\}}\sum_{y_{t+1}=0}^1f_{t+1}\lp[\bm{h}^{(t/(p-1))'},~y_{t+1}]'\rp\rrp.
\end{align*}
\end{proof} 

\bibliographystyle{apa}
\bibliography{refs}  

\end{document}